\documentclass[12pt,a4paper]{amsart}

\usepackage{amsmath,amsfonts,amssymb}

\usepackage[hmargin=2cm,vmargin=2cm]{geometry}

\usepackage[hypertexnames=false,hyperfootnotes=false,colorlinks=true,linkcolor=blue,%
citecolor=purple,filecolor=magenta,urlcolor=cyan,unicode,linktocpage=true,pagebackref=false]{hyperref}
\usepackage{nameref,zref-xr}

\usepackage{verbatim}     

\usepackage{tikz}

\newcommand{\mbZ}{\mathbb Z}
\newcommand{\mbC}{\mathbb C}
\newcommand{\cP}{\mathcal P}
\newcommand{\oM}{\overline{\mathcal M}}

\def\cM{{\mathcal{M}}}
\def\oM{{\overline{\mathcal{M}}}}
\def\CP{{{\mathbb C}{\mathbb P}}}

\def\d{{\partial}}

\newcommand{\cA}{\mathcal A}

\newcommand{\cF}{\mathcal F}
\renewcommand{\th}{\widetilde h}

\newcommand{\Coef}{\mathrm{Coef}}

\DeclareMathOperator{\Aut}{Aut}

\newcommand{\mbCP}{\mathbb{CP}}
\newcommand{\Cr}{\mathrm{Cr}}
\newcommand{\lb}{\left(}
\newcommand{\rb}{\right)}

\newcommand{\res}{\mathrm{res}}

\newtheorem{theorem}{Theorem}[section]
\newtheorem{proposition}[theorem]{Proposition}
\newtheorem{lemma}[theorem]{Lemma}

\theoremstyle{definition}

\newtheorem{definition}[theorem]{Definition}
\newtheorem{example}[theorem]{Example}
\newtheorem{remark}[theorem]{Remark}

\usepackage{color}

\def\&{\vspace{-5pt}&}

\newcommand{\mcP}{\mathcal P}

\newcommand{\bft}{\mathbf{t}}
\newcommand{\bfp}{\mathbf{p}}
\newcommand{\bff}{\mathbf{f}}
\newcommand{\wtL}{\widetilde{L}}
\newcommand{\whL}{\widehat{L}}
\newcommand{\bPhi}{\bar{\Phi}}
\newcommand{\bD}{\bar{D}}
\newcommand{\bfa}{\mathbf{a}}
\newcommand{\mbN}{\mathbb{N}}
\newcommand{\ttau}{\widetilde{\tau}}
\newcommand{\tf}{\widetilde{f}}
\newcommand{\tC}{\widetilde{C}}

\newcommand{\tK}{\widetilde{K}}
\newcommand{\tq}{\widetilde{q}}
\newcommand{\HH}{\mathrm{H}}

\newcommand{\Gr}{\mathrm{Gr}}

\DeclareMathOperator{\coker}{coker}
\DeclareMathOperator{\spn}{span}

\DeclareFontFamily{U}{mathx}{}
\DeclareFontShape{U}{mathx}{m}{n}{<-> mathx10}{}
\DeclareSymbolFont{mathx}{U}{mathx}{m}{n}
\DeclareMathAccent{\widehat}{0}{mathx}{"70}
\DeclareMathAccent{\widecheck}{0}{mathx}{"71}

\numberwithin{theorem}{section}
\numberwithin{equation}{section}

\begin{document}

\title{Open Hurwitz numbers and the mKP hierarchy}

\author{Alexandr Buryak}
\address{A. Buryak:\newline 
Faculty of Mathematics, National Research University Higher School of Economics, Usacheva str. 6, Moscow, 119048, Russian Federation;\smallskip\newline 
Skolkovo Institute of Science and Technology, Bolshoy Boulevard 30, bld. 1, Moscow, 121205, Russian Federation}
\email{aburyak@hse.ru}

\author{Ran J. Tessler}
\address{R.~J.~Tessler:\newline Department of Mathematics, Weizmann Institute of Science, POB 26, Rehovot 7610001, Israel}
\email{ran.tessler@weizmann.ac.il}

\author{Mikhail Troshkin}
\address{M. Troshkin:\newline 
Faculty of Mathematics, National Research University Higher School of Economics, Usacheva str. 6, Moscow, 119048, Russian Federation}
\email{mntroshkin@gmail.com}

\begin{abstract}
We give a natural definition of open Hurwitz numbers, where the weight of each ramified covering includes an integer parameter~$N$ taken to the power that is equal to the number of boundary components of a Riemann surface with boundary mapping to $\CP^1$. We prove that the resulting sequence of partition functions, depending on~$N\in\mbZ$, is a tau-sequence of the mKP hierarchy, or in other words it is a sequence of tau-functions of the KP hierarchy where each tau-function is obtained from the previous one by a B\"acklund--Darboux transformation. Our result is motivated by a previous observation of Alexandrov and the first two authors that the refined intersection numbers on the moduli spaces of Riemann surfaces with boundary give a tau-sequence of the mKP hierarchy.
\end{abstract}

\date{\today}

\maketitle

\section{Introduction}

The appearance of integrable systems in the study of curve-counting invariants is a well-known phenomenon, which was first manifested by Witten's conjecture~\cite{Wit91}, claiming that the generating series~$F^c$ of the intersection numbers of psi-classes on the moduli spaces of stable algebraic curves $\oM_{g,n}$ is the logarithm of a tau-function of the KP hierarchy. Witten's conjecture was proved by Kontsevich~\cite{Kon92}, and it inspired many similar results for Gromov--Witten invariants (see, e.g.,~\cite{DZ04,OP06,MST16}), Fan--Jarvis--Ruan--Witten invariants (see, e.g.,~\cite{Wit93,FJR13}), and more generally for the correlators of cohomological field theories (see, e.g.,~\cite{Kaz09,Bur15,LYZZ21,DZ00,BS12,BS24,LWZ25}). A new direction was opened by Okounkov~\cite{Oko00}, who proved that the generating series~$H^c$ of simple Hurwitz numbers is the logarithm of a tau-function of the KP hierarchy. This motivated an extensive study of the integrability of Hurwitz numbers of various types~(see, e.g., \cite{KL15,Har16,BDKS22}.

\medskip

More recently, studying the integrability of curve-counting invariants for Riemann surfaces with boundary became very active. The paper~\cite{PST24} initiated a study of the intersection theory on a compactification $\oM_{g,k,l}$ of the moduli space of Riemann surfaces with (nonempty) boundary $\cM_{g,k,l}$, where $g$ is the genus of the double a Riemann surface with boundary, $k$ is the number of boundary marked points, and $l$ is the number of internal marked points. The authors of~\cite{PST24} proposed an open (this adjective is often used to indicate that Riemann surfaces with boundary are considered) analog of Witten's conjecture, claiming that the exponent $e^{F^o}$ of the generating series $F^o$ of the intersection numbers on $\oM_{g,k,l}$ satisfies a certain system of PDEs, which was shown in~\cite{Bur16} to be equivalent to the system of equations for the wave function of the KP hierarchy. Based on the work~\cite{Tes23}, the Pandharipande--Solomon--Tessler conjecture was proved in~\cite{BT17}. An important clarification of the structure of $e^{F^o}$ was obtained in~\cite{Ale15}, where the author proved that $e^{F^c+F^o}$ is a tau-function of the KP hierarchy, which is related to $e^{F^c}$ by a B\"acklund--Darboux transformation.

\medskip

In the paper~\cite{ABT17}, the authors considered a refinement of the open intersection numbers by distinguishing contributions from surfaces with different numbers of boundary components. To be more precise, there is a decomposition $\oM_{g,k,l}=\bigsqcup_{1\le b\le g+1}\oM_{g,k,l,b}$, where $\oM_{g,k,l,b}$ parameterizes surfaces with boundary with $b$ boundary components. In~\cite{ABT17}, the authors considered the intersection numbers on each $\oM_{g,k,l,b}$, assigned the weight $N^b$ to them, and formed the generating series denoted by $F^{o,N}$, $N\in\mbZ$. Then the authors of~\cite{ABT17} conjectured that $e^{F^c+F^{o,N}}$ coincides with the tau-function given by the Kontsevich--Penner matrix model, which, according to~\cite{Ale15}, gives a tau-sequence of the modified KP (mKP) hierarchy.

\medskip

Our motivation was to give a Hurwitz-type analog of the observation from~\cite{ABT17}, i.e., to propose a natural definition of open Hurwitz numbers, counting maps from Riemann surfaces with boundary to $\CP^1$, that would give a tau-sequence of the mKP hierarchy including the KP tau-function $e^{H^c}$. This is done in this paper, which is organized as follows. After a brief review, in Section~\ref{section:closed Hurwitz numbers}, of simple Hurwitz numbers, which we call closed Hurwitz numbers, we move on, in Section~\ref{section:open Hurwitz numbers}, to our definition of open Hurwitz numbers counting certain maps from Riemann surfaces with boundary to $\CP^1$. Very similarly to what was done in~\cite{ABT17}, the weight of each map includes an integer parameter~$N$ taken to the power that is equal to the number of boundary components of a Riemann surface with boundary mapping to $\CP^1$. Then we prove Theorem~\ref{theorem:open Hurwitz and BD} claiming that, after a simple rescaling, the sequence of the exponents $\tau^o_N$ of the generating series of open Hurwitz numbers is a tau-sequence of the mKP hierarchy. Also, in Theorem~\ref{theorem:BD from tauzero to tauone}, we describe explicitly the B\"acklund--Darboux transformation relating $\tau^o_0$ and $\tau^o_1$. In order to make the paper more self-contained, we review in Appendix~\ref{section:kp} and Appendix~\ref{section:bd} the necessary facts about the KP hierarchy and the B\"acklund--Darboux transformations.

\medskip

\noindent{\bf Acknowledgements.} A.~B. is grateful to M. Kazarian for various useful discussions regarding the KP hierarchy and Hurwitz numbers. The work of A.~B and M.~T. is an output of a research project implemented as part of the Basic Research Program at the National Research University Higher School of Economics (HSE University). R.~T. was supported by ISF grant 1729/23.

\medskip

\section{Closed Hurwitz numbers}\label{section:closed Hurwitz numbers}

For $n\in\mbZ_{\ge 0}$, we denote by~$\cP_n$ the set of all partitions~$\lambda$ of~$n$. For a partition~$\lambda=(\lambda_1,\ldots,\lambda_l)$, $\lambda_1\ge\ldots\ge\lambda_l\ge 1$, we denote $l(\lambda):=l$.

\medskip

\begin{definition}
Let $d\in \mbZ_{\ge 1}$, $m\in\mbZ_{\ge 0},\lambda\in\mcP_d$. Let us fix pairwise distinct points $x_1,\ldots,x_m\in\mbC\subset\CP^1$.
\begin{enumerate}
\item Consider a connected closed Riemann surface $C$. A \emph{closed simple ramified covering} of degree~$d$ is a nonconstant holomorphic map $f\colon C\to\mbCP^1$ such that $x_1,\ldots,x_m$ are simple critical values of $f$, the ramification profile over~$\infty$ is given by the partition $\lambda$, and the map $f$ is not ramified over $\mbC\backslash\{x_1,\ldots,x_m\}$.

\smallskip

\item A \emph{closed Hurwitz number} $h^c(\lambda,m)$ is the number of isomorphism classes $[f]$ of closed simple ramified coverings counted with the weight $\frac{1}{\sharp\Aut(f)}$, where $\Aut(f)$ is the automorphism group of the covering. 
\end{enumerate}
\end{definition}

\medskip

Introduce formal variables $p_1,p_2,\ldots$, $\beta$, $q$, and consider the generating series 
$$
H^c(\bfp,\beta,q):=\sum_{\substack{d\ge 1\\m\ge 0}}\sum_{\lambda\in\mcP_d} h^c(\lambda,m)\frac{\beta^m}{m!}q^d p_\lambda\in\mbC[[p_1,p_2,\ldots,\beta,q]],
$$
where $p_\lambda:=\prod_{i=1}^{l(\lambda)}p_{\lambda_i}$ and we denote by $\bfp$ the sequence of variables $(p_1,p_2,\ldots)$. Define a \emph{closed partition function} $\tau^c$ by 
$$
\tau^c:=e^{H^c}.
$$

\medskip

It is well known that $\tau^c$ is a tau-function of the KP hierarchy (this was proved in~\cite{Oko00}, but see also~\cite{KL07} for a nice presentation), where the variables $p_i$ are related to the times $t_i$ of the KP hierarchy (see Appendix~\ref{section:kp} for a brief review) by $p_i=it_i$. Moreover, the tau-function can be explicitly described using the infinite wedge formalism in the following way:
\begin{gather}\label{eq:Fock and Hurwitz}
\tau^c=\Psi_0\Big(\bigwedge_{i\ge 1}f^\HH_i(z,\beta,q)\Big),\quad\text{where}\quad f^\HH_i(z,\beta,q):=z^{-i}\left(\sum_{l\ge 0}e^{\beta\frac{(l-i+\frac{1}{2})^2-(i-\frac{1}{2})^2}{2}}\frac{(qz)^l}{l!}\right),\quad i\ge 1.
\end{gather}

\medskip

\begin{remark}
In literature, tau-functions of the KP hierarchy are usually elements of the algebra of formal power series $\mbC[[\bft]]:=\mbC[[t_1,t_2,\ldots]]$. However, in many constructions and results,~$\mbC$ can be replaced by an arbitrary commutative associative $\mbC$-algebra $K$ (see Remark~\ref{remark:about base algebra} for more details). In this case, we will sometimes say that the \emph{base algebra} is $K$. So, $\tau^c$ can be considered as a family of KP tau-functions depending on the complex parameters $\beta$ and $q$ or as a tau-function of the KP hierarchy where the base algebra is $\mbC[[\beta,q]]$ (a bit more precisely, we have $\tau^c\in\mbC[e^\beta,e^{-\beta},q][[\bft]]$). 
\end{remark}

\medskip

%%%%%%%%%%%%%%%%%%%%%%%%%%%%%%%%%%%%%%%%%%%%%%%%%%%%%%%%%%%%%%%%%%%%%%%%%%%%
%%%%%%%%%%%%%%%%%%%%%%%%%%%%%%%%%%%%%%%%%%%%%%%%%%%%%%%%%%%%%%%%%%%%%%%%%%%%

\section{Open Hurwitz numbers}\label{section:open Hurwitz numbers}

\subsection{Definition}

We view $S^1$ as the unit circle in the complex plane~$\mbC$. Let us define the \emph{north} and the \emph{south} hemisphere of the complex projective line $\CP^1=\mbC\cup\{\infty\}$ by
$$
\CP^1_+:=\{z\in\mbC||z|<1\},\qquad \CP^1_-:=\{z\in\mbC||z|>1\}\cup\{\infty\}.
$$
Denote $\overline{\CP^1_+}:=\{z\in\mbC||z|\le 1\}$ and $\overline{\CP^1_-}:=\{z\in\mbC||z|\ge 1\}\cup\{\infty\}$. We will identify the relative homology group $H_2(\CP^1,S^1;\mbZ)$ with $\mbZ\oplus\mbZ$.

\medskip

\begin{definition}
Let $m_1,m_2\in\mbZ_{\ge 0}$ and $(d_1,d_2)\in\mbZ_{\ge 0}^2\backslash\{0\}$. Let us fix~$m_1$ pairwise distinct points in $\CP^1_+$ and~$m_2$ pairwise distinct points in $\CP^1_-\backslash\{\infty\}$ and denote the union of these points by $\Cr^s\subset\CP^1$. Let us fix a partition~$\lambda\in\mcP_{d_2}$.
\begin{enumerate}
\item Consider a connected compact Riemann surface with nonempty boundary $(C,\d C)$. An \emph{open simple ramified covering} of bidegree~$(d_1,d_2)$ is a nonconstant holomorphic map 
$$
f\colon(C,\d C)\to (\mbCP^1,S^1)
$$
satisfying the following properties:
\begin{enumerate}
\item The points from the set $\Cr^s$ are simple critical values of $f$. 

\smallskip

\item The ramification profile of the map $f$ over $\infty$ is given by the partition $\lambda$.

\smallskip

\item The map $f$ is not ramified over $\mbC\backslash\Cr^s$.

\smallskip

\item We have $f_*([C])=(d_1,d_2)\in\mbZ\oplus\mbZ=H_2(\CP^1,S^1;\mbZ)$, where $[C]\in H_2(C,\d C;\mbZ)$ is the fundamental class of $C$.
\end{enumerate}
We will say that a component of the boundary $\d C$ is \emph{negative} (respectively, \emph{positive}) if the image of a small neighbourhood of this component belongs to the south (respectively, north) hemisphere. See an illustration in Fig.~\ref{fig:open covering}. 

\smallskip

\item Let $N\in\mbZ$. An \emph{open Hurwitz number} $h^o_N(\lambda,m_1,m_2,d_1)$ is the number of isomorphism classes $[f]$ of open simple ramified coverings counted with the weight
$$
\frac{(-1)^{\sharp\left\{\begin{minipage}{2.4cm}\tiny negative\\ components of $\d C$\end{minipage}\right\}}N^{\sharp\left\{\begin{minipage}{2.4cm}\tiny components of $\d C$\end{minipage}\right\}}}{\sharp\Aut(f)}.
$$
\end{enumerate}
\end{definition}

\begin{figure}[t]
\centering
\begin{tikzpicture}[scale=0.6]

\draw[->] (0,3.75) to (0,2.25);
\node[left] at (0,3) {$f$};
\draw (5,0) node {$\bullet$}; \node[right] at (5,0) {$\infty$};
\draw (-3.5,0) node {$\bullet$}; \draw (-2.5,0) node {$\ldots$};\draw (-1.5,0) node {$\bullet$};
\draw (-2.5,-1.8	) node {$\underbrace{\phantom{aaaaaaa.}}_{\begin{minipage}{1.9cm}\tiny $m_1$ simple\\critical values\end{minipage}}$};
\draw (3.5,0) node {$\bullet$}; \draw (2.5,0) node {$\ldots$};\draw (1.5,0) node {$\bullet$};
\draw (2.5,-1.8	) node {$\underbrace{\phantom{aaaaaaa.}}_{\begin{minipage}{1.9cm}\tiny $m_2$ simple\\critical values\end{minipage}}$};
\draw (-4,1.5) node {$\CP^1_+$}; \draw (4,1.5) node {$\CP^1_-$};
\draw (0,0) ellipse (5 and 1);
\draw (0,1) arc(90:270:0.5cm and 1cm);
\draw[dashed] (0,-1) arc(-90:90:0.5cm and 1cm);

\begin{scope}[shift={(0,6)}]
\draw (0,1) arc(90:270:0.5cm and 1cm);
\draw (0,-1) arc(-90:90:0.5cm and 1cm);
\end{scope}

\begin{scope}[shift={(0,13.7)}]
\draw (0,1) arc(90:270:0.5cm and 1cm);
\draw[dashed] (0,-1) arc(-90:90:0.5cm and 1cm);
\end{scope}

\draw (0,10.8) to (-4,10.8) to[out=180,in=90] (-5,9.5) to[out=270,in=90] (-4.5,7.9) to[out=270,in=90] (-5,6.3) to[out=270,in=180] (-4,5) to (0,5);
\draw (0,8.8) to[out=180,in=180] (0,7);

\begin{scope}[shift={(0,-0.3)}]
\draw (0,15) to (4.3,15) to[out=0,in=90] (5,14.5) to[out=-90,in=90] (4.3,13.6) to[out=-90,in=90] (5,12.8) to[out=-90,in=90] (4.3,12) to[out=-90,in=90] (5,11.2) to[out=-90,in=90] (4.3,10.4) to[out=-90,in=90] (5,9.6) to[out=-90,in=0] (4.3,9.1) to (0,9.1);
\end{scope}

\draw (0,10.8) to[out=0,in=0] (0,12.7);

\begin{scope}[shift={(-3.7,9.5)},scale=0.4]
\draw (0,0) to [out=-35,in=-145] (4,0);
\draw (0.6499,-0.37) to [out=35,in=145] (4-0.6499,-0.37);
\end{scope}

\begin{scope}[shift={(-2.5,6.5)},scale=0.4]
\draw (0,0) to [out=-35,in=-145] (4,0);
\draw (0.6499,-0.37) to [out=35,in=145] (4-0.6499,-0.37);
\end{scope}

\begin{scope}[shift={(1.7,12.5)},scale=0.4]
\draw (0,0) to [out=-35,in=-145] (4,0);
\draw (0.6499,-0.37) to [out=35,in=145] (4-0.6499,-0.37);
\end{scope}

\draw (5,14.2) node {$\bullet$}; \node[right] at (5,14.2) {$\lambda_1$};
\draw (5,12.5) node {$\bullet$}; \node[right] at (5,11.9) {$\vdots$};
\draw (5,10.9) node {$\bullet$};
\draw (5,9.3) node {$\bullet$}; \node[right] at (5,9.3) {$\lambda_l$};
\draw (-1.5,13.7) node {\begin{minipage}{2cm}\tiny positive \\ boundary \\ component\end{minipage}};
\draw (2.5,6) node {\begin{minipage}{2cm}\tiny negative \\ boundary \\ component\end{minipage}};
\end{tikzpicture}
\caption{Open simple ramified covering}
\label{fig:open covering}
\end{figure}
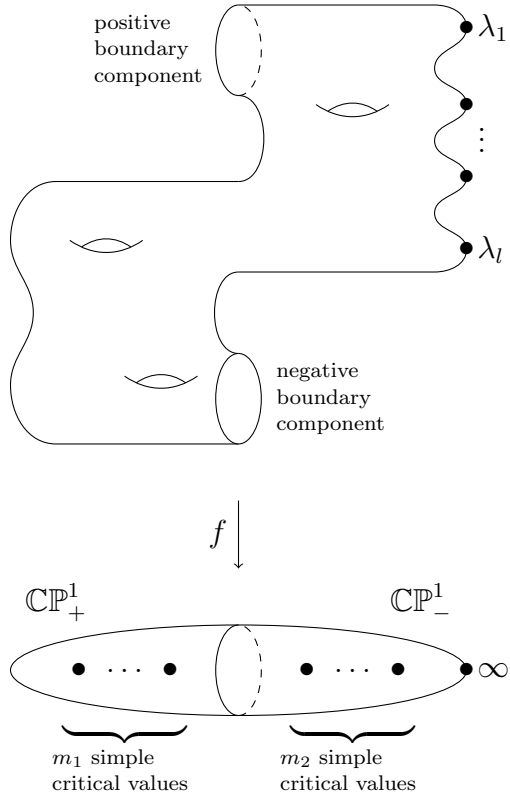

\medskip

\subsection{Open partition function as an mKP tau-sequence}

Let us introduce formal variables $q_1,q_2,\beta_1,\beta_2$ and form the following generating series:
$$
H^o_N(\bfp,\beta_1,\beta_2,q_1,q_2):=\sum_{\substack{(d_1,d_2)\in\mbZ_{\ge 0}^2\backslash\{0\}\\m_1,m_2\ge 0}}\sum_{\lambda\in\mcP_{d_2}}h^o_N(\lambda,m_1,m_2,d_1)\frac{\beta_1^{m_1}}{m_1!}\frac{\beta_2^{m_2}}{m_2!}q_1^{d_1}q_2^{d_2}p_\lambda.
$$
Define an \emph{open partition function} $\tau^o_N$ by
$$
\tau^o_N(\bfp,\beta_1,\beta_2,q_1,q_2):=e^{H^o_N(\bfp,\beta_1,\beta_2,q_1,q_2)+H^c(\bfp,\beta_1+\beta_2,q_1q_2)}.
$$
Note that 
\begin{gather}\label{eq:open tau-function for N=0}
\tau^o_0(\bfp,\beta_1,\beta_2,q_1,q_2)=\tau^c(\bfp,\beta_1+\beta_2,q_1q_2).
\end{gather}

\medskip

\begin{theorem}\label{theorem:open Hurwitz and BD}
Define $\ttau_N:=\left.\tau^o_N\right|_{q_2\mapsto e^{N\beta_2}q_2}$.
Then the sequence $(\ttau_N)_{N\in\mbZ}$ is a tau-sequence of the mKP hierarchy.
\end{theorem}

In other words, for each $N\in\mbZ$ we have the following:
\begin{enumerate}
	\item $\ttau_N$ is a tau-function of the KP hierarchy.  
	
	\smallskip
	
	\item $\ttau_{N+1}$ is related to $\ttau_N$ by a forward B\"acklund--Darboux transformation (see Appendix~\ref{section:bd} for a brief review).
\end{enumerate}

\begin{proof}
The proof will proceed in four steps.
\begin{enumerate}
	\item Lemma \ref{lemma:tau-o-beta2-zero} gives a formula for $\tau^o_N|_{\beta_2=0}$ in terms of $\tau^c$.
	
\smallskip	
	
	\item Proposition \ref{prop:tau-sequence-from-shifts} for $\tau = \tau^c$ implies that $(\tau^o_N|_{\beta_2=0})_{N\in\mbZ}$ is an mKP tau-sequence.
	
\smallskip	
	
	\item By Lemma \ref{lemma:cut-and-join-for-tau-o}, the full series $\tau^o_N$ can recovered as $\tau^o_N=e^{\beta_2 \cA}\lb\tau^o_N|_{\beta_2=0}\rb$, where $\cA$ is the classical cut-and-join operator for closed Hurwitz numbers.
	
\smallskip	
	
	\item Finally, Proposition \ref{proposition:cut-and-join and mKP} implies that $(\ttau_N)_{N\in\mbZ}$ is an mKP tau-sequence, which completes the proof.
\end{enumerate}

\smallskip

\begin{lemma}
\label{lemma:tau-o-beta2-zero}
We have
\begin{gather}\label{eq:tauc with beta2 zero}
\left.\tau^o_N\right|_{\beta_2=0}=\tau^c(\bfp-N[q_2^{-1}],\beta_1,q_1q_2)e^{N\sum_{n\ge 1}\frac{p_n}{n}q_2^n},
\end{gather}
where $\bfp-N[q_2^{-1}]:=\lb p_1-\frac{N}{q_2},p_2-\frac{N}{q_2^2},\ldots\rb$.
\end{lemma}
\begin{proof}
Suppose that $f\colon C \to\CP^1$ is a closed simple ramified covering of degree~$d$ and without critical values in $\CP^1_-\backslash\{\infty\}$. Let $f^{-1}(\infty)=\{p_1,\ldots,p_l\}$ and the multiplicity of $f$ at $p_i$ is equal to $\lambda_i$. Denote by $C_i$, $1\le i\le l$, the connected component of $f^{-1}(\CP^1_-)$ containing $p_i$. Since the map $f$ doesn't have critical values in $\CP^1_-\backslash\{\infty\}$, we have $C_i\cap C_j=\emptyset$ for $i\ne j$, and there exist biholomorphisms $\phi_i\colon\CP^1_-\to C_i$ such that $(f\circ\phi_i)(z)=z^{\lambda_i}$. One can remove some of the sets $C_1,\ldots,C_l$, say $C_{i_1},\ldots,C_{i_k}$, and obtain an open simple ramified covering of bidegree $(d,d-\sum_{j=1}^k\lambda_{i_j})$ with only negative boundary components. Let us call this procedure the \emph{cutting procedure}, see Fig.~\ref{fig:cutting procedure}.
\begin{figure}[t]
\centering
\begin{tikzpicture}[scale=0.5]

\draw[->] (0,3.75) to (0,2.25);
\draw (5,0) node {$\bullet$}; \node[right] at (5,0) {$\infty$};
\draw (-3.5,0) node {$\bullet$}; \draw (-2.5,0) node {$\ldots$};\draw (-1.5,0) node {$\bullet$};
\draw (-2.5,-1.8	) node {$\underbrace{\phantom{aaaaaaa.}}_{\begin{minipage}{1.9cm}\tiny $m_1$ simple\\critical values\end{minipage}}$};

\draw (-4,1.5) node {$\CP^1_+$}; \draw (4,1.5) node {$\CP^1_-$};
\draw (0,0) ellipse (5 and 1);
\draw (0,1) arc(90:270:0.5cm and 1cm);
\draw[dashed] (0,-1) arc(-90:90:0.5cm and 1cm);

\begin{scope}[shift={(0,6)}]
\draw (0,1) arc(90:270:0.5cm and 1cm);
\draw[dashed] (0,-1) arc(-90:90:0.5cm and 1cm);
\end{scope}
\begin{scope}[shift={(0,9.8)}]
\draw (0,1) arc(90:270:0.5cm and 1cm);
\draw[dashed] (0,-1) arc(-90:90:0.5cm and 1cm);
\end{scope}

\draw (0,5) arc(-90:90:5cm and 1cm);
\draw (0,8.8) arc(-90:90:5cm and 1cm);
\draw (5,9.8) node {$\bullet$};
\draw (5,6) node {$\bullet$};

\draw (0,10.8) to (-4,10.8) to[out=180,in=90] (-5,9.5) to[out=270,in=90] (-4.5,7.9) to[out=270,in=90] (-5,6.3) to[out=270,in=180] (-4,5) to (0,5);
\draw (0,8.8) to[out=180,in=180] (0,7);

\begin{scope}[shift={(-3.7,9.5)},scale=0.4]
\draw (0,0) to [out=-35,in=-145] (4,0);
\draw (0.6499,-0.37) to [out=35,in=145] (4-0.6499,-0.37);
\end{scope}

\begin{scope}[shift={(-2.5,6.5)},scale=0.4]
\draw (0,0) to [out=-35,in=-145] (4,0);
\draw (0.6499,-0.37) to [out=35,in=145] (4-0.6499,-0.37);
\end{scope}

\draw[->,thick] (7.6,5) to (9.6,5);

\begin{scope}[shift={(17,0)}]
\draw[->] (0,3.75) to (0,2.25);
\draw (5,0) node {$\bullet$}; \node[right] at (5,0) {$\infty$};
\draw (-3.5,0) node {$\bullet$}; \draw (-2.5,0) node {$\ldots$};\draw (-1.5,0) node {$\bullet$};
\draw (-2.5,-1.8	) node {$\underbrace{\phantom{aaaaaaa.}}_{\begin{minipage}{1.9cm}\tiny $m_1$ simple\\critical values\end{minipage}}$};

\draw (-4,1.5) node {$\CP^1_+$}; \draw (4,1.5) node {$\CP^1_-$};
\draw (0,0) ellipse (5 and 1);
\draw (0,1) arc(90:270:0.5cm and 1cm);
\draw[dashed] (0,-1) arc(-90:90:0.5cm and 1cm);

\begin{scope}[shift={(0,6)}]
\draw (0,1) arc(90:270:0.5cm and 1cm);
\draw (0,-1) arc(-90:90:0.5cm and 1cm);
\end{scope}
\begin{scope}[shift={(0,9.8)}]
\draw (0,1) arc(90:270:0.5cm and 1cm);
\draw[dashed] (0,-1) arc(-90:90:0.5cm and 1cm);
\end{scope}

\draw (0,8.8) arc(-90:90:5cm and 1cm);
\draw (5,9.8) node {$\bullet$};

\draw (0,10.8) to (-4,10.8) to[out=180,in=90] (-5,9.5) to[out=270,in=90] (-4.5,7.9) to[out=270,in=90] (-5,6.3) to[out=270,in=180] (-4,5) to (0,5);
\draw (0,8.8) to[out=180,in=180] (0,7);

\begin{scope}[shift={(-3.7,9.5)},scale=0.4]
\draw (0,0) to [out=-35,in=-145] (4,0);
\draw (0.6499,-0.37) to [out=35,in=145] (4-0.6499,-0.37);
\end{scope}

\begin{scope}[shift={(-2.5,6.5)},scale=0.4]
\draw (0,0) to [out=-35,in=-145] (4,0);
\draw (0.6499,-0.37) to [out=35,in=145] (4-0.6499,-0.37);
\end{scope}
\end{scope}
\end{tikzpicture}
\caption{Cutting procedure}
\label{fig:cutting procedure}
\end{figure}
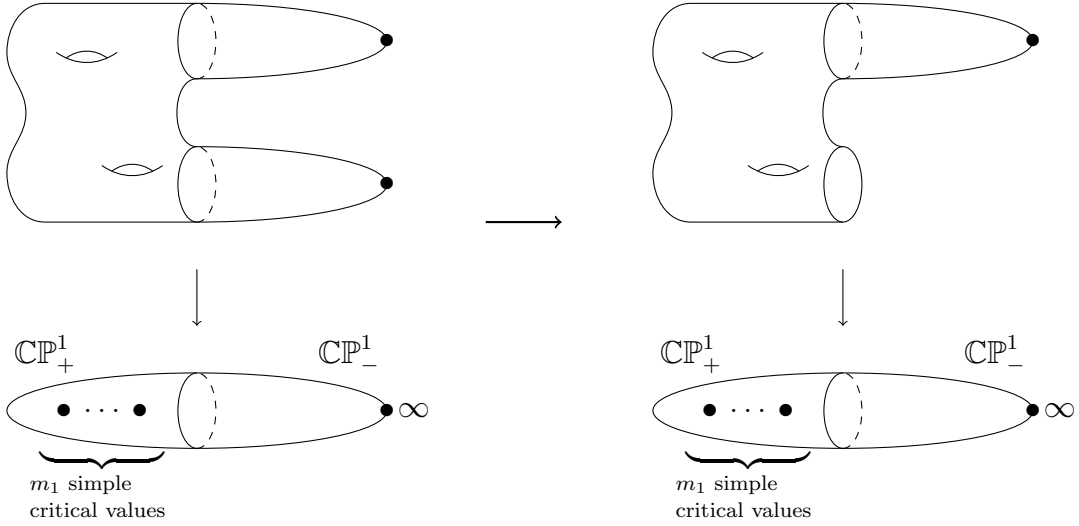
Conversely, if $f\colon (C,\d C)\to (\CP^1,S^1)$ is an open simple ramified covering with only negative boundary components, then one can glue to~$C$ closed disks along the components of~$\d C$ and obtain a closed simple ramified covering $\tf\colon\tC\to\CP^1$ such that $f$ is obtained from $\tf$ by the cutting procedure.

\medskip

Consider an open simple ramified covering $f\colon (C,\d C)\to (\CP^1,S^1)$ with $m_2=0$. There are two cases.

\medskip

\underline{\it Case 1: $f^{-1}(\CP^1_+)=\emptyset$}. Then one can easily see that $l(\lambda)=1$ and $f$ is isomorphic to the open simple ramified covering $(\overline{\CP^1_-},S^1)\to (\overline{\CP^1_-},S^1)$ mapping $z$ to $z^{\lambda_1}$. Therefore, $h^o_N(\lambda_1,0,0,0)=\frac{N}{\lambda_1}$. This gives the term $e^{N\sum_{n\ge 1}\frac{p_n}{n}q_2^n}$ on the right-hand side of~\eqref{eq:tauc with beta2 zero}. 

\medskip

\underline{\it Case 2: $f^{-1}(\CP^1_+)\ne\emptyset$}. Then $C$ has only negative boundary components and is obtained from some closed simple ramified covering by the cutting procedure. This gives the term $\tau^c(\bfp-N[q_2^{-1}],\beta_1,q_1q_2)$ on the right-hand side of~\eqref{eq:tauc with beta2 zero}.
\end{proof}

\medskip

\begin{proposition}
	\label{prop:tau-sequence-from-shifts}
	Consider an arbitrary tau-function $\tau(\bft)$ of the KP hierarchy, with a base algebra $K$. Let $\tau_N(\bft, x) := \tau(\bft - N[x^{-1}])e^{N\xi(\bft, x)}$, $N\in\mbZ$. Then $(\tau_N(\bft,x))_{N \in \mbZ}$ is a tau-sequence of the mKP hierarchy, with the base algebra $K((x^{-1}))$.
\end{proposition}
\begin{proof}
	Let us fix $N\ge 1$. Introduce a \emph{shift operator} $G_x$ by $G_x(f):=f|_{\bft \mapsto \bft - [x^{-1}]}$. Consider the group $K[[x_1^\pm,\ldots,x_N^\pm]]$ and denote by $K_N$ its subgroup defined as the sum of the subgroups of the form
	$$
	K[[x_i^{-1}]]_{i\in I}[x_j]_{j\in[N]\backslash I},\quad I\subset[N]:=\{1,\ldots,N\}.
	$$
	Clearly, the product of the elements from $K_N$ is well defined, so $K_N$ is an algebra. We have the inclusion of algebras 
	$$
	K_N\subset \underbrace{K((x_N^{-1}))((x_{N-1}^{-1}))\ldots ((x_1^{-1}))}_{\tK_N:=}.
	$$
	Note that for any element $f\in K_N$ the substitution $f|_{x_i\mapsto x}$ is well defined giving an element from $K((x^{-1}))$.
	
	\medskip
	
	We have $G_{x_1} G_{x_2} \dots G_{x_N} \tau \times \prod_{1 \leq i \leq N} e^{\xi(\bft, x_i)}\in K_N[[\bft]]$ and
	\begin{equation*}
		\tau_N = \tau(\bft - N[x^{-1}])e^{N\xi(\bft, x)} = 	\Big( G_{x_1} G_{x_2} \dots G_{x_N} \tau \times\prod_{1 \leq i \leq N} e^{\xi(\bft, x_i)} \Big) \Big|_{x_i \mapsto x}.
	\end{equation*}
	Also, we have $G_{x_1} e^{\xi(\bft, x_2)} =  \frac{x_1 - x_2}{x_1} e^{\xi(\bft, x_2)}\in K((x_2^{-1}))((x_1^{-1}))[[\bft]]\subset\tK_N[[\bft]]$. Therefore,
	\begin{multline}\label{eq:two N-shifts}
		G_{x_1} G_{x_2} \dots G_{x_N} \tau  \, \times \!\!\prod_{1 \leq i \leq N}\!\! e^{\xi(\bft, x_i)} =\\
		= \underbrace{G_{x_1} \left( G_{x_2} \left(\dots \left(G_{x_N} \tau \cdot e^{\xi(\bft, x_N)}\right) \dots\right) \cdot e^{\xi(\bft, x_2)} \right) \cdot e^{\xi(\bft, x_1)}}_{\in\tK_N[[\bft]]}\, \times \!\! \underbrace{\prod_{1\leq i < j \leq N} \frac{x_i}{x_i - x_j}}_{\in\tK_N}.
	\end{multline}
	
	\medskip
	
	Without loss of generality, we can assume that $\tau(0)=1$. We know that our tau-function~$\tau$ has a fermionic representation
	$$
	\tau(\bft) = \Psi_0 \Big(\bigwedge_{i\le -1} f_i(z)\Big),\quad f_i(z)=z^{i}\Big(1+\sum_{j\ge 1}f_{i,j}z^j\Big),\quad f_{i,j}\in K.
	$$	
	By Proposition \ref{proposition:wave-fermionic}, the right-hand side of equation~\eqref{eq:two N-shifts} can be represented as
	\begin{multline*}
		\Psi_N \Big(  x_1^{-N+1}\theta(x_1) \, x_2^{-N+2}\theta(x_2) \dots \theta(x_N) \bigwedge f_i \Big) \, \times \!\! \prod_{1\leq i < j \leq N} \frac{x_i}{x_i - x_j} =\\
		= \Psi_N \Big(\theta(x_1)\theta(x_2) \dots \theta(x_N) \bigwedge f_i \Big) \, \times \!\! \prod_{1\leq i < j \leq N} \frac{1}{x_i - x_j},
	\end{multline*}
	where we view $\theta(x_1)\theta(x_2) \dots \theta(x_N) \bigwedge f_i$ as an element of the Fock space $\cF^{[N]}_{\tK_N}$. 
	
	\medskip
	
	Consider integers $a_1 > a_2 > \dots > a_N$. The coefficient of $\theta_{a_1} \theta_{a_2} \dots \theta_{a_N}$ in the operator $\theta(x_1) \theta(x_2) \dots \theta(x_N) \prod_{1\leq i < j \leq N} \frac{1}{x_i - x_j}$ is equal to the Laurent polynomial
	\begin{equation*}
		s_{\lambda}(x_1, \dots, x_N) = \frac{ \det \left( x_i^{a_j} \right)_{i,j=1}^N}{\prod_{1\leq i < j \leq N} (x_i - x_j)}, \quad \lambda_i = a_i+ i - N, \quad i=1,\dots, N.
	\end{equation*}
	The specialization to $x_i = x$ is found by the following standard computation: 
	\begin{multline*}
		s_\lambda(x, \dots, x) = x^{|\lambda|}\left.s_\lambda(1, q, \dots, q^{N-1})\right|_{q=1}= x^{|\lambda|}\left.\frac{ \det \left(q^{(i-1)a_j} \right)}{\prod_{i < j} (q^{i-1} - q^{j-1})}\right|_{q=1}=\\
		=x^{|\lambda|}\left.\prod_{i<j}\frac{q^{a_j}-q^{a_i}}{q^{i-1} - q^{j-1}}\right|_{q=1}=x^{|\lambda|}\prod_{i < j} \frac{a_i - a_j}{j - i} =  x^{|\lambda|}\frac{\det \left(a_i^{N-j}\right)}{\prod_{i=1}^N(i-1)!}.
	\end{multline*}
	This is the coefficient of $\theta_{a_1} \theta_{a_2} \dots \theta_{a_N}$ in
	\begin{equation*}
		\prod_{i=1}^{N} \frac{x^{-i + 1}}{(i-1)!} \times \big( (x\d_x)^{N-1} \theta(x) \big) \circ \dots \circ \big( (x \d_x) \theta(x) \big) \circ \theta(x).
	\end{equation*}
	We conclude that 
	\begin{multline*}
		\tau_N=\Psi_N\bigg(\bigg[\prod_{i=1}^{N} \frac{x^{-i + 1}}{(i-1)!} \times \big( (x\d_x)^{N-1} \theta(x) \big) \circ \dots \circ \big( (x \d_x) \theta(x) \big) \circ \theta(x)\bigg]\bigwedge_{i\le -1} f_i\bigg)=\\
		=\Psi_N\bigg(h_{N-1}\wedge h_{N-2}\wedge\ldots\wedge h_0\wedge\bigwedge_{i\le -1} f_i\bigg),
	\end{multline*}
	where $h_j=\sum_{a\in\mbZ}\frac{(a+j)^j}{j!}x^a z^{a+j}$, $0\le j\le N-1$. 
	
	\medskip
	
	We can decompose $h_0=\sum_{a\in\mbZ}x^a z^a$ in the following way:
	$$
	h_0=\sum_{i\le -1}c_{0,i}f_i+\th_0,\quad \th_0\in K((x^{-1}))[[z]],\quad c_{0,i}\in x^i\mbC[[x^{-1}]].
	$$
	Therefore, $\tau_1=\Psi_1\bigg(\th_0\wedge\bigwedge_{i\le -1} f_i\bigg)$, and we see that $\Coef_{z^0}\th_0=\tau(-[x^{-1}])\in K[[x^{-1}]]$. Continuing this process, we construct a sequence of elements $\th_i\in z^i K((x^{-1}))[[z]]$, $i\ge 0$, such that
	$$
	h_i=\sum_{j\le -1}c_{i,j}f_j+\sum_{k=0}^{i-1} d_{i,k}\th_k+\th_i,\quad c_{i,j}\in x^{j-i}K[[x^{-1}]],\quad d_{i,k}\in K((x^{-1})),
	$$
	and moreover $\Coef_{z^i}\th_i=\frac{\tau(-(i+1)[x^{-1}])}{\tau(-i[x^{-1}])}$. 
	
	\medskip
	
	We have $\tau_M=\Psi_M\bigg(\th_{M-1}\wedge\th_{M-2}\wedge\ldots\wedge \th_0\wedge\bigwedge_{i\le -1} f_i\bigg)$, $M\ge 1$, and therefore $\tau_{M}$ is obtained from $\tau_{M-1}$ by a forward B\"acklund--Darboux transformation. Since $\tau_{-M}=((\tau^*)_M)^*$, we see that $\tau_{-M}$ is obtained from $\tau_{-(M-1)}$ by a backward B\"acklund--Darboux transformation, and thus~$\tau_{-(M-1)}$ is obtained from~$\tau_{-M}$ by a forward B\"acklund--Darboux transformation, which completes the proof of the proposition.  
\end{proof}

\smallskip

\begin{lemma}
\label{lemma:cut-and-join-for-tau-o}
Define the (classical) \emph{cut-and-join operator} $\cA$ as:
\begin{equation*}
	\cA:=\frac{1}{2}\sum_{i,j\ge 1}\left((i+j)p_ip_j\frac{\d}{\d p_{i+j}}+ijp_{i+j}\frac{\d^2}{\d p_i\d p_j}\right).
\end{equation*}
Then we have $\frac{\d\tau^o_N}{\d\beta_2}=\cA \tau^o_N$.
\end{lemma}
\begin{proof}
Taking into account~\eqref{eq:open tau-function for N=0}, for $N=0$ this becomes the famous cut-and-join equation for the closed Hurwitz numbers, see the paper~\cite{GJ97}, where the authors gave a proof using the interpretation of the closed Hurwitz numbers as the number of factorisations of a permutation into transpositions. The cut-and-join equation for the closed Hurwitz numbers has also a geometric proof, see e.g.~\cite[Section~3]{EMS11}. One can easily see that exactly the same arguments work in the case of arbitrary~$N$.
\end{proof} 

\smallskip

\begin{proposition}\label{proposition:cut-and-join and mKP}
	Let $\bff(z, q) = (f_k(z, q))_{k \in \mbZ}$ be a sequence of Laurent series of the form
	\begin{equation*}
		f_k(z, q) = z^k \left( 1 + \sum_{j \geq 1} f_{k,j} (q z)^j \right) \in \mbC[[q]]((z)), \quad f_{k,j} \in \mbC.
	\end{equation*}
	Let $(\tau_N)_{N \in \mbZ}$ be the corresponding mKP tau-sequence given by \eqref{eq:mkp-sequence-from-series}.	 Then the sequence $(\ttau_N(\beta))_{N \in \mbZ}$, where $\ttau_N(\beta):=\left. (e^{\beta \cA}\tau_N) \right|_{q \mapsto e^{N \beta} q}$, is also an mKP tau-sequence.
\end{proposition}
\begin{proof}
By Proposition~\ref{proposition:cut-and-join and Fock}, we have $\ttau_{N}(\beta) = \Psi_N \left( \bigwedge_{i \geqslant 1} \tilde{f}^{(N)}_{N-i}(z, q, \beta) \right)$, where
	\begin{gather*}
		\tilde{f}^{(N)}_{k}(z, q, \beta):=z^{k}\Big(1 +\sum_{j \geq 1}e^{\beta\frac{(j + k - N + \frac{1}{2})^2-(k -N +\frac{1}{2})^2}{2}} e^{\beta jN} f_{k,j} (qz)^j \Big).
	\end{gather*}
	We observe that $\tilde{f}^{(N)}_k(z,q, \beta)$ does not depend on $N$ by transforming the exponent of $e^{\beta}$: $\frac{(j + k - N + \frac{1}{2})^2-(k -N +\frac{1}{2})^2}{2} + jN  = \frac{(j + k + \frac{1}{2})^2-(k +\frac{1}{2})^2}{2}$, and therefore, by formula~\eqref{eq:mkp-sequence-from-series}, $(\ttau_N(\beta))_{N \in \mbZ}$ is an mKP tau-sequence, as desired.
\end{proof}
\end{proof}

\medskip

The next theorem describes explicitly the B\"acklund--Darboux transformation relating~$\ttau_0=\tau^c(\bfp,\beta_1+\beta_2,q_1q_2)$ and~$\ttau_1=\tau^o_1|_{q_2\mapsto e^{\beta_2}q_2}$. 

\begin{theorem}\label{theorem:BD from tauzero to tauone}
We have
$$
\tau^o_1=\Coef_{x^0}\left[D(x,\beta_1,\beta_2,q_1,q_2)\tau^c(\bfp-[x^{-1}],\beta_1+\beta_2,e^{-\beta_2}q_1q_2)e^{\sum_{n\ge 1}\frac{p_n}{n}x^n}\right],
$$
where
$$
D(z,\beta_1,\beta_2,q_1,q_2)=\sum_{l\ge 0}\lb 1+\sum_{k\ge 1}e^{\beta_1\frac{l^2+l-k^2+k}{2}}\frac{(-1)^k q_1^{k+l}}{(l+k)(k-1)!l!}\rb  e^{\beta_2\frac{l^2-l}{2}}\lb\frac{q_2}{z}\rb^l.
$$
\end{theorem} 
\begin{proof}
Using Propositions~\ref{proposition:wave-fermionic},~\ref{proposition:cut-and-join and Fock}, and formula~\eqref{eq:Fock and Hurwitz}, we compute
\begin{align*}
\tau^o_1=&e^{\beta_2 \cA}\lb\tau^c(\bfp-[q_2^{-1}],\beta_1,q_1q_2)e^{\sum_{n\ge 1}\frac{p_n}{n}q_2^n}\rb=\\
=&e^{\beta_2 \cA}\Psi_1\lb\sum_{i\in \mbZ}q_2^i z^i\wedge\bigwedge_{j\ge 1}f^\HH_j(z,\beta_1,q_1q_2)\rb=\\
=&e^{(\beta_1+\beta_2)\cA}\Psi_1\lb\sum_{i\in \mbZ}q_2^i e^{-\beta_1\frac{i^2-i}{2}} z^i\wedge\bigwedge_{j\ge 1}f^\HH_j(z,0,e^{\beta_1}q_1q_2)\rb=\\
=&e^{(\beta_1+\beta_2)\cA}\Psi_1\lb\sum_{i\in \mbZ}q_2^i e^{-\beta_1\frac{i^2-i}{2}} z^i\wedge\bigwedge_{j\ge 1} z^{-j} e^{\tq_1 q_2 z}\rb,
\end{align*}
where $\tq_1:=e^{\beta_1}q_1$

\medskip

\begin{lemma}
For $k \geq 1$, we have 
$$
z^{-k}\wedge\bigwedge_{j\ge 1} z^{-j} e^{q z}=(-q)^k\bigg( \sum_{l \geq 0}  \frac{q^{l}}{(l+k)(k-1)!l!} z^l \bigg) \wedge \bigwedge_{j\ge 1} z^{-j} e^{q z}.
$$
\end{lemma}
\begin{proof}
We have
\begin{multline*}
z^{-k} = e^{-q z} (e^{q z} z^{-k})= \sum_{m \geq 0} \frac{(-q)^m}{m!} \cdot e^{q z}z^{m-k}=\\
= \sum_{0\le m<k} \frac{(-q)^m}{m!} \cdot e^{q z}z^{m-k} + (-q)^k\sum_{m \geq 0} \frac{(-q)^{m}}{(m+k)!}e^{q z} z^{m}.
\end{multline*}
The first sum vanishes in the wedge product, and we transform the second sum:
\begin{equation*}
\sum_{m \geq 0} \frac{(-q)^{m}}{(m+k)!}e^{q z} z^{m} = \sum_{m, n \geq 0} \frac{(-q)^{m}}{(m+k!)} \cdot \frac{q^n}{n!} z^{m+n}= \sum_{l \geq 0} \left( \sum_{0 \leq m \leq l} \frac{(-1)^{m}}{(k+m)!(l - m)!} \right) q^{l} z^l.
\end{equation*}
The sum in the parenthesis can be simplified via the binomial formula:
\begin{multline*}
\sum_{0 \leq m \leq l}\frac{(-1)^m}{(k+m)!(l - m)!} = \frac{1}{(l+k)!} \sum_{0 \leq m \leq l} (-1)^{m} \binom{l+k}{k+m} = \\
= \frac{1}{(l+k)!} \binom{l+k-1}{k-1}= \frac{1}{(l+k)(k-1)!l!},
\end{multline*}
which completes the proof of the lemma.
\end{proof}

\medskip

So we obtain
\begin{align*}
\tau^o_1=&e^{(\beta_1+\beta_2)\cA}\Psi_1\lb\sum_{l\ge 0}\lb e^{-\beta_1\frac{l^2-l}{2}}+\sum_{k\ge 1}e^{-\beta_1\frac{k^2+k}{2}}\frac{(-1)^k\tq_1^{k+l}}{(l+k)(k-1)!l!}\rb (q_2 z)^l\wedge\bigwedge_{j\ge 1} z^{-j} e^{\tq_1 q_2 z}\rb=\\
=&\Psi_1\lb\sum_{l\ge 0}\lb 1+\sum_{k\ge 1}e^{\beta_1\frac{l^2+l-k^2+k}{2}}\frac{(-1)^k q_1^{k+l}}{(l+k)(k-1)!l!}\rb  e^{\beta_2\frac{l^2-l}{2}} (q_2 z)^l\wedge\bigwedge_{j\ge 1} f^\HH_j(z,\beta_1+\beta_2,e^{-\beta_2} q_1 q_2)\rb,
\end{align*}
which completes the proof of the theorem.
\end{proof}

\medskip

{\appendix	

\section{KP hierarchy}\label{section:kp}

\subsection{KP hierarchy in the Lax form}
\label{subsection:kp-lax}

In this section, we recall the definition of the Kadomtsev--Petviashvili (KP) hierarchy in terms of formal pseudo-differential operators. For more detailed information, see, e.g., \cite{Dic03}.

\medskip

Let $\bft := (t_1, t_2, \dots)$ be an infinite sequence of variables and $\mbC[[\bft]] := \mbC[[t_1, t_2, \dots]]$. Let $\d=\d_1 := \frac{\d}{\d t_1}$ and $\d_n := \frac{\d}{\d t_n}$ for $n>1$ be the partial derivatives acting on the algebra $\mbC[[\bft]]$.  A \emph{formal pseudo-differential operator} is an expression of the form 
$$
A = \sum_{i=-N}^{\infty} a_i \d^{-i}, \quad a_i \in \mbC[[\bft]].
$$	
We define multiplication of two such expressions using the following generalization of the Leibniz rule: $\d^n \circ a:=\sum_{i=0}^\infty \binom{n}{i}(\d^i a) \circ \d^{n-i}$, $n$, $a \in \mbC[[\bft]]$, where $\binom{n}{i} := \frac{1}{i!} \cdot n(n-1)\cdots (n-i+1)$ for any $n \in \mbZ$. For a pseudo-differential operator $A$ as above, we define the residue $\res\,A:=a_{-1}$ and the decomposition $A = A_+ + A_-$: $A_+:=\sum_{i \leq 0} a_i \d^{-i}$, $A_-:=\sum_{i\geq 1} a_i \d^{-i}$. We use the shorthand $A^n_+ := (A^n)_+$, and similarly, $A^n_- := (A^n)_-$.

\medskip

For a sequence $(u_1(\bft), u_2(\bft), \dots)$, $u_i(\bft) \in \mbC[[\bft]]$, let the pseudo-differential operator $L(\bft)$ be given~by
\begin{equation*}
	L(\bft) = \d + \sum_{i = 1}^{\infty} u_i(\bft) \d^{-i}.
\end{equation*}

\medskip

\begin{definition}
A sequence $(u_i(\bft))$, or the corresponding pseudo-differential operator $L(\bft)$, is a solution of the \emph{KP hierarchy} if, for each $n > 1$, the operator $L(\bft)$ satisfies
	\begin{equation*}
		\d_n L(\bft) = [(L(\bft)^n)_+,\, L(\bft)],
	\end{equation*}
	where derivatives $\d_n$ act on $L(\bft)$ coefficient-wise in $\d$.
\end{definition}

\medskip

\begin{proposition}
	For each solution $L = L(\bft)$ of the KP hierarchy, there exists a pseudo-differential operator $P$, called a \emph{dressing operator}, of the form $P = 1 + \sum_{i = 1}^\infty w_i(\bft) \d^{-i}$, such that $L = P \circ \d \circ P^{-1}$ and $\d_n P=-L^n_-\circ P$. Such operator $P$ is defined up to \emph{gauge transformations} of the form $P \longmapsto P \circ \left( 1 + \sum_{i=1}^\infty c_i \d^{-i}\right)$, where $(c_i)$ is an arbitrary sequence of constant coefficients.
\end{proposition}

\medskip

Let $\xi(\bft, x): = \sum_{i\geq 1} t_i x^i$. To each solution $L$ of the KP hierarchy and its dressing operator~$P$, we associate the \emph{wave function} $\psi(\bft, x):=Pe^{\xi(\bft, x)}$. Explicitly, $\psi(\bft, x) = \left( 1 + \sum_{i=1}^\infty w_i(\bft)z^{-i} \right) e^{\xi(\bft, x)}$. 

\medskip

\begin{proposition}\label{proposition:wave-is-eigenfunction}
The wave function satisfies the following equations for all $n \geq 1$: $\d_n \psi = L^n_+ \psi$.
\end{proposition}

\medskip

It will be convenient to introduce additional dressing operators and wave functions \emph{of level $N \in \mbZ$} by the following formulas: $P_N: = P \circ \d^N$, $\psi_N(\bft, x) := x^N \psi(\bft, x)$.

\medskip

The KP hierarchy has the following remarkable property: even though \emph{a priori} its solution is given in terms of an infinite sequence of unknown functions $u_i(\bft)$, in fact, each solution may be described by a single function of $\bft$, called the \emph{tau-function} $\tau(\bft)$. Let $\bft - [x^{-1}]$ be the shifted sequence of variables $\left(t_1 - x^{-1}, \, t_2 - \frac{1}{2}x^{-2}, \, t_3 - \frac{1}{3}x^{-3}, \dots \right)$.

\medskip

\begin{theorem}\label{theorem:tau-functions}
	For each solution $L(\bft)$ of the KP hierarchy and its associated wave function~$\psi(\bft, x)$, there exists a function $\tau(\bft)$ that satisfies
	\begin{equation*}
		\psi(\bft, x) = \frac{\tau(\bft - [x^{-1}])}{\tau(\bft)} e^{\xi(\bft, x)}, \quad \tau(0) \ne 0.
	\end{equation*}
	The function $\tau$ is unique up to the multiplication $\tau\mapsto\lambda\tau$, $\lambda\ne 0$. It is called a \emph{tau-function} corresponding to the solution $L(\bft)$.
\end{theorem} 

\medskip

The action of gauge transformations of the dressing operator on the wave and tau-function can be written as follows:
\begin{equation*}
	P \longmapsto P \circ \exp\!\left(\sum_{i\ge 1} d_i \d^{-i}\right), \quad \psi \longmapsto \psi \exp \! \left(\sum_{i\ge 1} d_i x^{-i} \right), \quad \tau \longmapsto  \tau \exp \! \left(-\sum_{i\ge 1} id_i t_i \right),
\end{equation*}
where $\sum d_i x^i = \log(1 + \sum c_i x^i)$.

\medskip

\subsection{Adjoint operators}\label{subsection:adjoint}

For a pseudo-differential operator $A = \sum a_i \d^i$, we define its formal adjoint by $A^\dagger:=\sum (-\d)^i \circ a_i$. Adjoint operators satisfy the following identities:
\begin{enumerate}
	\item $(A \circ B)^\dagger = B^\dagger \circ A^\dagger$, hence $[A, B]^\dagger = -[A^\dagger, B^\dagger]$;

\smallskip	
	
	\item $(A_+)^\dagger = (A^\dagger)_+$ and $(A_-)^\dagger = (A^\dagger)_-$;

\smallskip

	\item $(A^\dagger)^\dagger = A$.
\end{enumerate}

\medskip

Given a KP solution $L$ with a dressing operator $P$, we define its \emph{adjoint wave function} $\psi^\dagger(\bft, x):=(P^\dagger)^{-1} e^{-\xi(\bft, x)}$. It satisfies a system of linear equations analogous to those in Proposition \ref{proposition:wave-is-eigenfunction}: $\d_n \psi^\dagger = -(L^n_+)^\dagger \psi^\dagger$. The adjoint wave function can be expressed in terms of the tau-function as
\begin{equation*}
	\psi^\dagger(\bft, x) = \frac{\tau(\bft + [x^{-1}])}{\tau(\bft)} e^{-\xi(\bft, x)},
\end{equation*}
where $\bft + [x^{-1}] := (t_1 + x^{-1}, t_2 + \frac{1}{2}x^{-2}, t_3 + \frac{1}{3}x^{-3}, \dots)$. The adjoint wave function at level $N$ is introduced as $	\psi^\dagger_N(\bft, x) := x^{-N} \psi^\dagger(\bft, x)$.

\medskip

For a pseudo-differential operator $A = \sum a_i(\bft)\d^i$, define $A^!:=\sum a_i(-\bft) (-\d)^i$. This transformation satisfies the following properties:
\begin{equation*}
 (A \circ B)^! = A^!\circ B^!,\qquad (A_\pm)^! = A^!_\pm,\qquad (A^!)^! = A,\qquad (A^\dagger)^!=(A^!)^\dagger.
\end{equation*}

\medskip

Let $L = \d + \sum u_i(\bft) \d^{-i}$ be a solution of the KP hierarchy. Define 
\begin{equation*}
	L^*(\bft): = (L^!)^\dagger.
\end{equation*}
Then $L^*$ also satisfies the KP equations. This operation defines an involution on the set of all KP solutions, which we call the \emph{adjunction involution}. It acts on the dressing operator and on the wave and tau-function by the following formulas:
\begin{equation*}
	P^*=((P^\dagger)^{-1})^!,\qquad\psi^*(\bft, x) = \psi^\dagger(-\bft, x), \qquad \tau^*(\bft) = \tau(-\bft).
\end{equation*}

\medskip

\subsection{Fock space formalism}\label{subsection:Fock}

We begin by recalling the notion of the fermionic Fock space, also known as the semi-infinite wedge space $\Lambda^{\!\frac{\infty}{2}} \mbC((z))$.

\medskip

A sequence of integers $\bfa = (a_1, a_2, \dots)$ is called \emph{admissible of level $N$} if it satisfies $a_i = -i + N$ for all sufficiently large $i$. The \emph{degree} of such a sequence is defined as $\deg(\bfa) := \sum_i (a_i + i - N)$. To any admissible sequence $\bfa$, we associate the infinite Grassmann monomial $z^\bfa := z^{a_1} \wedge z^{a_2} \wedge \cdots$, subject to the following identification:  $z^\bfa = (-1)^{|\sigma|} z^{\sigma(\bfa)}$ for any finitely supported permutation $\sigma \colon \mbN \to \mbN$. In particular, if $a_i = a_j$ for some $i \neq j$, then $z^{\bfa} = 0$. It follows that if $\deg(\bfa) < 0$, then $z^\bfa = 0$.

\medskip  

Let $\cF^{[N]}_d$ denote the vector space spanned by all monomials $z^\bfa$ of level $N$ and degree $d$. The \emph{fermionic Fock space of level $N$}, denoted $\cF^{[N]}$, consists of all formal sums $w_0 + w_1 + \cdots, \quad \text{where } w_d \in \cF^{[N]}_d$. The \emph{full fermionic Fock space} is then defined as the direct sum $\cF := \bigoplus_{N \in \mathbb{Z}} \cF^{[N]}$.

\medskip

A \emph{partition} $\lambda$ is a non-increasing sequence of non-negative integers $\lambda_1 \geq \lambda_2 \geq \dots$, where $\lambda_i = 0$ for $i \gg 0$. The space $\cF^{[N]}_d$ has a basis $\{v_\lambda^{[N]}\}$, $v_\lambda^{[N]}:=z^{\lambda_1 - 1 + N} \wedge z^{\lambda_2 - 2 + N} \wedge \dots$, where $\lambda$ runs over all partitions of weight $d$. Denote the monomial $v^{[N]}_0 = z^{N-1} \wedge z^{N-2} \wedge \dots$ by~$\left| N \right \rangle$ (the \emph{vacuum vector of level~$N$}). It is the unique monomial of the lowest possible degree~$0$ in $\cF^{[N]}$. We define the dot product $\left \langle \cdot \mid \cdot \right \rangle$ on $\cF^{[N]}_{d}$ by declaring $\{v_\lambda^{[N]}\}$ to be an orthonormal basis.

\medskip

For $a \in \mbZ$, define linear operators $\theta_a\colon \cF^{[N]}_d \to \cF^{[N+1]}_{d+a-N}$ and $\theta_a^\dagger\colon \cF^{[N]}_d \to \cF^{[N-1]}_{d-a+N-1}$ as follows:
\begin{align*}
	\theta_a(z^{a_1} \wedge z^{a_2} \wedge \dots)&:=z^a \wedge z^{a_1} \wedge z^{a_2} \wedge \dots ,\\
	\theta^\dagger_a(z^{a_1} \wedge z^{a_2} \wedge \dots)&:=\sum_{j \geq 1} \delta_{a, a_j} (-1)^{j-1} z^{a_1} \wedge \dots \wedge z^{a_{j-1}} \wedge \widehat{z^{a_j}} \wedge z^{a_{j+1}} \wedge \dots .
\end{align*}
We also define the generating series of these operators: $\theta(x):=\sum_{i\in\mbZ}\theta_i x^i, \qquad \theta^\dagger(x) := \sum_{i\in\mbZ}\theta^\dagger_i x^{-i}$.

\medskip

For $0 \neq n \in \mbZ$, define $\alpha_n\colon \cF^{[N]}_d \to \cF^{[N]}_{d-n}$ as $\alpha_n: = \sum_{i \in \mbZ} \theta_i \circ \theta^\dagger_{i+n}$. A direct computation shows that $\alpha_n$ act on the monomials $z^\bfa$ as follows:
\begin{equation*}
	\alpha_n(z^{a_1} \wedge z^{a_2} \wedge \dots) = \sum_{j\ge 1} z^{a_1} \wedge \dots \wedge z^{a_{j}-n} \wedge \dots.
\end{equation*}
Define $\Gamma(\bft):=\sum_{n\ge 1} t_n\alpha_n$.

\medskip

The \emph{bosonic Fock space} is the space of formal power series $\mbC[[\bft]]$. The \emph{bosonic-fermionic correspondence at level $N$} is a linear operator $\Psi_N : \cF^{[N]} \to \mbC[[\bft]]$ defined by the following formula:
\begin{equation*}
	\Psi_N(u) = \left \langle N \mid e^{\Gamma(\bft)} u \right \rangle, \quad u \in \cF^{[N]}.
\end{equation*}
The operator $\Psi_N$ is an isomorphism between $\cF^{[N]}$ and $\mbC[[\bft]]$.

\medskip

The bosonic-fermionic correspondence will be used in Theorem \ref{theorem:Fock and KP} to give a geometric intepretation of KP tau-functions in terms of decomposable Fock space vectors, or, equivalently, elements of the infinite-dimensional Grassmannian.

\begin{comment}
\medskip

Note that Part 2 of the theorem implies that for any $D(x)=\sum_{i \geq N} d_i x^{-i}\in\mbC((x^{-1}))$ satisfying $d_{N}=1$ we have
\begin{equation*}
	\Coef_{x^0}\left[D(x)\psi_{N}(\bft,x) \tau_{\bff}(\bft) \right]=\Psi_{N+1} \Big( \sum_{i \geq N} d_i \theta_i \bigwedge_{i\ge 1}f_i(z)\Big)
	=\Psi_{N+1} \Big(D(z^{-1})\wedge\bigwedge_{i\ge 1}f_i(z)\Big).
\end{equation*}
This implies that $\Coef_{x^0}\left[D(x)\psi_{N}(\bft,x) \tau_{\bff}(\bft) \right]$ is a tau-function of the KP hierarchy corresponding to the sequence $\left( D(z^{-1}), f_1(z), f_2(z),\ldots\right)$.

\medskip
\end{comment}

\medskip

\subsection{Sato Grassmannian}

\label{subsection:grassmannian}

Let $\mbC((z))$ be the field of formal Laurent series. As a vector space over $\mbC$, it admits a direct sum decomposition $\mbC((z)) = \mbC[[z]] \oplus z^{-1}\mbC[z^{-1}]$. Let $\pi\colon \mbC((z)) \to z^{-1} \mbC[z^{-1}]$ be the corresponding projection along $\mbC[[z]]$.

\medskip

\begin{definition}
The \emph{Sato Grassmannian} $\Gr$ is the set of all $\mbC$-linear subspaces $H$ in $\mbC((z))$ such that the restriction $\pi|_H\colon H \to z^{-1}\mbC[z^{-1}]$ has finite-dimensional kernel and cokernel. 
\end{definition}

\medskip

$\Gr$ is the disjoint union of its components $\Gr_N$, defined as follows:
\begin{equation*}
	\Gr_N := \{H \in \Gr \mid \dim \ker \pi|_H - \dim \coker \pi|_H = N\}.
\end{equation*}
The \emph{big cell} $\Gr_N^0$ is the set of subspaces $H$ such that the projection $\pi_N|_H$ from $H$ to $z^{N-1}\mbC[z^{-1}]$ along $z^N \mbC[[z]]$ is an isomorphism.
Each subspace $H \in \Gr^0_N$ has an \emph{adapted basis} $(f_{N-1}, f_{N-2}, \dots)$ of the following form:
$f_k(z) = (\pi_N|_H)^{-1}\big(z^k\big) = z^k \Big(1 + \sum_{j \geq 0} f_{k, j} z^{j}\Big)$, $k < N$.

\medskip

Given an adapted basis of $H \in \Gr_N^0$, we can construct a decomposable vector $\bigwedge_i f_{N-i}(z)$ in the Fock space $\cF^{[N]}$. The relation of decomposable vectors to the KP hierarchy is revealed in the following theorem. 

\begin{theorem}[\cite{DKJM83, MJD00}]\label{theorem:Fock and KP}
For some $N \in \mbZ$, consider a sequence of Laurent series $\bff=(f_{N-1}(z), f_{N-2}(z),\ldots)$, such that for all $k < N$, $f_{k}(z) \in \mbC((z))$ has the following form:
$$
f_{k}(z)=z^{k} \Big(1 + \sum_{j \geq 1} f_{k,j} z^j\Big), \quad f_{k,j} \in \mbC.
$$
Then the formal power series 
\begin{equation}\label{eq:tau-function from f-sequence}
\tau_{\bff} := \Psi_N \Big( \bigwedge_{i\ge 1}f_{N-i}(z)\Big) \in \mbC[[\bft]]
\end{equation}
is a tau-function of the KP hierarchy. Moreover, an arbitrary tau-function $\tau$ of the KP hierarchy satisfying $\tau(0)=1$ can be represented in the form~\eqref{eq:tau-function from f-sequence}.
\end{theorem}

Thus, via the bosonic-fermionic correspondence, from each subspace $H \in \Gr^0_N$ we obtain a tau-function of the KP hierarchy $\tau_H(\bft)$.

\medskip

Let $(-,-)$ be an inner product on $\mbC((z))$ given by the formula $(f(z), g(z)) = \res_{z=0} f(z) g(z)$. Then, for any $H \in \Gr$, its orthogonal subspace $H^\perp$ is also an element of $\Gr$. Denote by $\iota$ the resulting map $\Gr \to \Gr, H \mapsto H^\perp$. The map $\iota$ is an involution and satisfies $\iota(\Gr_N) = \Gr_{-N}$ and $\iota(\Gr^0_N) = \Gr^0_{-N}$.

\medskip

\begin{proposition}[{\cite[{Section 4.14}]{HB21}}]
Let $H \in \Gr^0_N$ and $H^\perp = \iota(H) \in \Gr^0_{-N}$. Then the corresponding tau-functions are related as follows: $\tau_{H^\perp}(\bft) = \tau_H(-\bft)$. 
\end{proposition}
 
Thus, the orthogonality involution $\iota$ acts on KP tau-functions as the adjunction involution introduced in \ref{subsection:adjoint}: $\tau_{H^\perp} = (\tau_H)^*$.

\medskip

The wave and the adjoint wave functions also admit fermionic representations, given in the following proposition.

\begin{proposition}[{\cite{MJD00}}]
	\label{proposition:wave-fermionic}
	Let $\tau_{\bff} = \Psi_N \Big( \bigwedge_{i\ge 1}f_{N-i}(z)\Big)$. Then the wave function $\psi_{N}(\bft, x)$ and the adjoint wave function $\psi^\dagger_{N}(\bft, x)$ at level $N$, associated to $\tau_{\bff}$, are given by
	\begin{gather}\label{eq:psi in terms of fi}
		\psi_{N}(\bft, x) = \frac{\Psi_{N+1} \Big( \theta(x) \bigwedge_{i\ge 1}f_{N-i}(z)\Big)}{\Psi_N \Big( \bigwedge_{i\ge 1}f_{N-i}(z)\Big) },
		\quad \psi^\dagger_{N}(\bft, x) = x^{-1}\frac{\Psi_{N-1} \Big( \theta^\dagger(x) \bigwedge_{i\ge 1}f_{N-i}(z)\Big)}{\Psi_N \Big( \bigwedge_{i\ge 1}f_{N-i}(z)\Big) }.
	\end{gather}
	In other words, the following equations hold:
	\begin{align*}
		\tau_\bff(\bft - [x^{-1}])e^{\xi(\bft, x)} &= x^{-N} \Psi_{N+1} \Big( \theta(x) \bigwedge f_{N-i}(z)\Big),\\
		 \tau_\bff(\bft + [x^{-1}])e^{-\xi(\bft, x)} &= x^{N-1} \Psi_{N-1} \Big( \theta^\dagger(x) \bigwedge f_{N-i}(z)\Big).
	\end{align*}
\end{proposition}

\medskip

\begin{remark}\label{remark:about base algebra}
All the results presented above have a direct generalization to the case when we start from the pseudo-differential operators $A = \sum_{i=-N}^{\infty} a_i \d^{-i}$ with $a_i\in K[[\bft]]=K[[t_1,t_2,\ldots]]$, where $K$ is an arbitrary commutative associative $\mbC$-algebra. In this generalization, one should require in Theorem~\ref{theorem:tau-functions} that $\tau(0)\in K$ and $\lambda\in K$ are invertible and also, in the Fock space formalism, one should consider $K$-modules instead of vector spaces over $\mbC$. In the case, when~$K$ is different from $\mbC$, we will sometimes say that the \emph{base algebra} is~$K$. The Fock space with the coefficients from $K$ will be denoted by $\cF_K$.
\end{remark}

\medskip

\subsection{Cut-and-join operator}

The \emph{cut-and-join operator} is defined by 
$$
\cA:=\frac{1}{2}\sum_{i,j\ge 1}\left(ijt_it_j\frac{\d}{\d t_{i+j}}+(i+j)t_{i+j}\frac{\d^2}{\d t_i\d t_j}\right).
$$

\begin{proposition}[see, e.g., \cite{KL07}]\label{proposition:cut-and-join and Fock}
Let $\tau\in\mbC[[\bft]]$ be an arbitrary tau-function of the KP hierarchy.
\begin{enumerate}
\item The formal power series $e^{\beta \cA}\tau\in\mbC[[\beta,\bft]]$ is also a tau-function of the KP hierarchy, with the base algebra $\mbC[[\beta]]$.
		
\smallskip
		
\item If $\tau$ is given by formula~\eqref{eq:tau-function from f-sequence} from Theorem \ref{theorem:Fock and KP}, then $e^{\beta \cA}\tau=\tau_{\bff(\beta)}$, where the sequence $\bff(\beta)=(f_{N-1}(z,\beta),f_{N-2}(z,\beta),\ldots)$ is given by 
$$
f_{k}(z,\beta):=z^{k}\Big(1 +\sum_{j \geq 1}e^{\beta\frac{(k + j - N + \frac{1}{2})^2-(k - N +\frac{1}{2})^2}{2}} f_{k,j}z^j \Big).
$$
\end{enumerate}
\end{proposition}

\medskip

\begin{example}
The famous cut-and-join equation for the closed Hurwitz numbers (see Section~\ref{section:closed Hurwitz numbers}) claims that $\frac{\d \tau^c}{\d \beta} = \cA \tau^c$. Since $\tau^c|_{\beta = 0} = e^{t_1}=\Psi_0(\bigwedge_{i \geq 1} f_{-i})$, where $f_{-i}(z) = e^z z^{-i}$, by Proposition \ref{proposition:cut-and-join and Fock}, $\tau^c$ is a tau-function of the KP hierarchy, and moreover $\tau^c$ is given by formula~\eqref{eq:Fock and Hurwitz}.
\end{example}

\medskip

\section{B\"acklund--Darboux transformations and mKP hierarchy}\label{section:bd}

\subsection{B\"acklund--Darboux transformations in the Lax form}

Here we discuss B\"acklund--Darboux transformations and their properties, following the results presented in \cite{CSY92, HvdL01, KvdL18}. In order to make the paper more self-contained, we will give short proofs of some statements.

\medskip

Let $L(\bft)$ be a solution of the KP hierarchy, and let $D(\bft)$ be an invertible pseudo-differential operator. We formulate conditions under which $\wtL = D\circ L\circ D^{-1}$ is also a solution of the KP hierarchy.

\begin{lemma}\label{lemma:sato-wilson}
If $D$ satisfies the following equations for all $n \geq 1$: $\d_n D = \wtL^n_+\circ D - D\circ L^n_+$, then $\wtL = D\circ L\circ D^{-1}$ is a solution of the KP hierarchy.
\end{lemma}

\begin{comment}
	\begin{proof}
		Consider the equation $\dd_j \wtL - [\wtL^j_+, \wtL] = 0$ and transform the left-hand side.
		\begin{align*}
			&\dd_j (DLD^{-1}) - [\wtL^j_+, DLD^{-1}] = \\
			&= \, \dd_j D \circ L D^{-1} + D\circ[L^j_+, L] \circ D^{-1} - D L D^{-1} \circ \dd_j D \circ D^{-1} - [\wtL^j_+, DLD^{-1}]\\
			&= \dd_j D \circ L D^{-1} - D L D^{-1} \circ \dd_j D \circ D^{-1} - (\wtL^j_+ D - D L^j_+)(L D^{-1}) + DLD^{-1}(\wtL^j_+ D - D L^j_+)D^{-1}.
		\end{align*}
		We see that upon the substitution $\dd_j D = \wtL^j_+ D - D L^j_+$, the equation is satisfied.
	\end{proof}
\end{comment}

\medskip

%	Note that
%	\begin{equation*}
	%			\wtL^n_+ D - D L^n_+ =  D L^n_- - \wtL^n_- D.
	%	\end{equation*}
%	If $D$ is a purely differential operator of order $n$, then this expression is also a differential operator (by the left-hand side) of order at most $n - 1$ (by the right-hand side). In particular, if $D$ is a differential operator of order $1$, this expression is simply a function.

\begin{definition}
Let $L(\bft)$ be a KP solution. An associated \emph{eigenfunction} of $L$ is a formal power series~$\Phi(\bft)\in A$ that satisfies the following equations for all $n\ge 1$: $\d_n \Phi = L^n_+ \Phi$.
\end{definition}

\medskip

\begin{proposition}\label{proposition:bd-definition}
Let $L$ be a KP solution, and let $\Phi$ be its associated eigenfunction satisfying $\Phi(0)\ne 0$. Define $D: = \Phi \circ \d \circ \Phi^{-1} = \d - \d \log\Phi$. 	Then $D$ satisfies the conditions of Lemma~\ref{lemma:sato-wilson}, and so $\whL := D\circ L\circ D^{-1}$ is again a KP solution.
\end{proposition}
\begin{proof}
Let us check that $\d_j D=\whL^j_+\circ D-D\circ L^j_+$. For this, we transform the right-hand side as follows: 
	\begin{multline*}
		\whL^j_+\circ D-D\circ L^j_+=(D\circ L^j\circ D^{-1})_+\circ D-D\circ L^j_+=\\
		=(D\circ L^j_+\circ D^{-1})_+\circ D-D\circ L^j_+=-(D\circ L^j_+\circ D^{-1})_-\circ D.
	\end{multline*}
We see that the pseudo-differential operator under consideration is actually a function, which is equal to
	\begin{multline*}
		-\res(D\circ L^j_+\circ D^{-1})=-\res(\Phi\circ\d\circ\Phi^{-1}\circ L^j_+\circ\Phi\circ\d^{-1}\circ\Phi^{-1})=\\
		=-\res(\d\circ\Phi^{-1}\circ L^j_+\circ\Phi\circ\d^{-1})=-\d\left(\Phi^{-1}L^j_+\Phi\right)=-\d(\d_j\log\Phi)=\d_j D,
	\end{multline*}
	as required.
\end{proof}

\medskip

\begin{definition}
	Given $L$ and $\whL$ as in Proposition \ref{proposition:bd-definition}, we say that $\whL$ is related to $L$ by a \emph{(forward) B\"acklund--Darboux transformation}.
\end{definition}

\medskip

Under a forward B\"acklund--Darboux transformation, the dressing operator and the wave function are transformed as follows: $\widehat{P}_{N+1} = D\circ P_{N}$, $\widehat{\psi}_{N + 1} = D\psi_{N}$.

\medskip

Proposition~\ref{proposition:wave-is-eigenfunction} implies that any infinite linear combination of the coefficients of~$x^i$ of the wave function $\psi(\bft, x)$ is an eigenfunction. The converse statement is also true.

\begin{proposition}\label{proposition:eigenfunctions-from-wave}
	Consider an arbitrary solution~$L$ of the KP hierarchy.
	\begin{enumerate}
		\item For any $k\in\mbZ$ and any formal power series $C(x) = \sum_{i \geq 0} c_i x^{-i}$ with constant coefficients, the formal power series
		\begin{equation}\label{eq:eigenfunction from psi}
			\Phi(\bft) = \Coef_{x^k} \left( C(x) \psi(\bft, x) \right),
		\end{equation}
		is an eigenfunction.
		
		\smallskip
		
		\item An arbitrary eigenfunction $\Phi$ has the form~\eqref{eq:eigenfunction from psi} with $k=0$. Moreover, the power series $C(x)$ is determined by an eigenfunction uniquely.
	\end{enumerate}
\end{proposition}
\begin{proof}
	{\it 1}. As we already mentioned above, this immediately follows from Proposition~\ref{proposition:wave-is-eigenfunction}.
	
	\medskip
	
{\it 2}. It suffices to check that any initial condition $\Phi|_{t_{\ge 2}=0}=:\sum_{i\ge 0}\phi_i t_1^i$ can be obtained as $\left.\Coef_{x^0}\left(C(x)\psi(\bft,x)\right)\right|_{t_{\ge 2}=0}$. Denote $1+\sum_{i\ge 1}p_i(t_1)\d^{-i}:=P|_{t_{\ge 2}=0}$. 	Then
	\begin{align*}
		\left.\Coef_{x^0}\left(C(x)\psi(\bft,x)\right)\right|_{t_{\ge 2}=0}=&\Coef_{x^0}\bigg(\bigg(\sum_{i\ge 0}c_i x^{-i}\bigg)\bigg(1+\sum_{j\ge 1}p_j(t_1)x^{-j}\bigg)e^{x t_1}\bigg)=\\
		=&\sum_{j\ge 0}\frac{c_jt_1^{j}}{j!}+\sum_{i\ge 1,\,j\ge 0}\frac{p_i(t_1)c_jt_1^{i+j}}{(i+j)!},
	\end{align*}
	and we recursively choose $c_k$ so that $	\frac{c_k}{k!}=\phi_k-\Coef_{t_1^k}\left(\sum_{i\ge 1,\,0\le j\le k-i}\frac{p_i(t_1)c_jt_1^{i+j}}{(i+j)!}\right)$, $k\ge 0$.
\end{proof}

\medskip

Note that the eigenfunction given by~\eqref{eq:eigenfunction from psi} with $k=0$ satisfies the property $\Phi(0)\ne 0$ if and only if $c_0\ne 0$.

\medskip

\begin{proposition}\label{proposition:backlund-for-tau}
	Under the B\"acklund--Darboux transformation generated by an eigenfunction~$\Phi$ satisfying $\Phi(0)\ne 0$, the tau-function is transformed as $\widehat{\tau}=\Phi\tau$.
\end{proposition}
\begin{proof}
	We have to check that $\widehat{\psi}=\frac{\Phi(\bft-[x^{-1}])}{\Phi(\bft)}\psi$, which is equivalent to the equation
	$$
	x^{-1}(\d\psi-\d\log\Phi\cdot\psi)=\frac{\Phi(\bft-[x^{-1}])}{\Phi(\bft)}\psi \quad \Leftrightarrow \quad \Phi(\bft-[x^{-1}])+x^{-1}\d\Phi=x^{-1}\d\log\psi\cdot\Phi.
	$$
	Since $\Phi$ has the form $\Phi(\bft) = \Coef_{z^0} \left( C(z) \psi(\bft, z) \right)$, the last equation follows from the equation
	\begin{multline*}
		\psi(\bft-[x^{-1}],z)+x^{-1}\d\psi(\bft,z)=x^{-1}\d\log\psi(\bft,x)\cdot \psi(\bft,z)\quad\Leftrightarrow\\
		\Leftrightarrow\quad  \frac{\psi(\bft-[x^{-1}],z)}{\psi(\bft,z)}=x^{-1}\left(\d\log\psi(\bft,x)-\d\log\psi(\bft,z)\right)\quad\Leftrightarrow\\
		\Leftrightarrow\quad \frac{\psi(\bft-[x^{-1}],z)}{\psi(\bft,z)}=x^{-1}\d\log\frac{\psi(\bft,x)}{\psi(\bft,z)}. 
	\end{multline*}
	In terms of tau-functions, the last equation is equivalent to
	$$
	\frac{\tau(\bft-[x^{-1}]-[z^{-1}])\tau(\bft)}{\tau(\bft-[x^{-1}])\tau(\bft-[z^{-1}])}(1-zx^{-1})=x^{-1}\left(\d\log\left(\frac{\tau(\bft-[x^{-1}])}{\tau(\bft-[z^{-1}])}\right)+x-z\right).
	$$
	Finally, this is equivalent to
	\begin{multline*}
		(x-z)\left(\tau(\bft-[x^{-1}]-[z^{-1}])\tau(\bft)-\tau(\bft-[x^{-1}])\tau(\bft-[z^{-1}])\right)=\\
		=\d\tau(\bft-[x^{-1}])\cdot \tau(\bft-[z^{-1}])-\d\tau(\bft-[z^{-1}])\cdot \tau(\bft-[x^{-1}]),
	\end{multline*}
	where we recognize the differential Fay identity \cite[Section~6.4.13]{Dic03} (one should substitute $x=s_1^{-1}$ and $z=s_2^{-1}$).
\end{proof}

\medskip

Combining Propositions \ref{proposition:eigenfunctions-from-wave} and \ref{proposition:backlund-for-tau}, we obtain the following formula:
\begin{equation}
	\label{eq:backlund-darboux-for-tau}
	\widehat{\tau}(\bft) = \Coef_{x^0} \left( C(x)\, \psi(\bft, x) \tau(\bft) \right) = \Coef_{x^0} \left( C(x)\, \tau(\bft - [x^{-1}])\, e^{\xi(\bft,x)} \right).
\end{equation}

\medskip

There is a dual notion of backward B\"acklund--Darboux transformations.

\begin{definition}
An \emph{adjoint eigenfunction} $\bPhi(\bft)$ for a KP solution $L(\bft)$ is a function which satisfies the following equations for all $n \geq 1$: $	\d_n \bPhi = -(L^n_+)^\dagger \bPhi$. 
\end{definition}

\medskip

Analogously to Proposition \ref{proposition:eigenfunctions-from-wave}, we have the representation of adjoint eigenfunctions in terms of the adjoint wave function:
\begin{equation}\label{eq:adjoint Phi in terms of adjoint psi}
	\bPhi(\bft) = \Coef_{x^k} \left( C(x) \psi^\dagger(\bft, x) \right),\quad k\in\mbZ,
\end{equation}
where $C(x)=\sum_{i\ge 0}c_i x^{-i}$ is a power series with constant coefficients.

\medskip

\begin{proposition}\label{proposition:backward bd}
	Let $L$ be a KP solution, and let $\bPhi$ be an associated adjoint eigenfunction satisfying $\bPhi(0)\ne 0$. Define $\bD:=\bPhi^{-1} \circ \d^{-1} \circ \bPhi$. Then we have the following:
\begin{enumerate}
		\item $\bD$ satisfies the conditions of Lemma~\ref{lemma:sato-wilson}, and so $\widecheck{L}: = \bD\circ L\circ \bD^{-1}$ is again a KP solution.  We say that $\widecheck{L}$ is related to $L$ by a \emph{backward B\"acklund--Darboux transformation}.
		
		\smallskip
		
		\item The action on the dressing operator, the wave function, and the tau-function is given by $\widecheck{P}_{N-1} = \bD \circ P_{N}$, $\widecheck{\psi}_{N - 1} = \bD \psi_{N}$, $\widecheck{\tau} = \bPhi \tau$.
\end{enumerate}
\end{proposition}
\begin{proof}
	{\it 1}. The conditions of Lemma~\ref{lemma:sato-wilson} are equivalent to $\d_n(\bD^{-1})=L^n_+\circ\bD^{-1}-\bD^{-1}\circ\widecheck{L}^n_+$, which can be checked analogously to the proof of Proposition~\ref{proposition:bd-definition}.
	
	\medskip
	
	{\it 2}. The formulas for $\widecheck{P}_{N-1}$ and $\widecheck{\psi}_{N - 1}$ are easily proved. The formula for $\widecheck{\tau}$ is equivalent to
	$$
	\widecheck{\psi}=\frac{\bPhi(\bft-[x^{-1}])}{\bPhi(\bft)}\psi\;\Leftrightarrow\; x\left(\bPhi^{-1}\circ\d^{-1}\circ\bPhi\right)\psi=\frac{\bPhi(\bft-[x^{-1}])}{\bPhi(\bft)}\psi\;\Leftrightarrow\; x\bPhi \psi=\d\left(\bPhi(\bft-[x^{-1}])\psi\right),
	$$ 
	which, using~\eqref{eq:adjoint Phi in terms of adjoint psi}, follows from the equation $x\psi^\dagger(\bft,z)\psi(\bft,x)=\d\left(\psi^\dagger(\bft-[x^{-1}],z)\psi(\bft,x)\right)$. 	Expressing $\psi$ and $\psi^\dagger$ in terms of the tau-function, this equation is derived from the differential Fay identity by a direct computation.
\end{proof}

\medskip

It follows that the backward B\"acklund--Darboux transformation of the tau-function can be expressed as
\begin{equation*}
	\widecheck{\tau}(\bft) = \Coef_{x^0} \left( C(x)\, \tau(\bft + [x^{-1}])\, e^{-\xi(\bft,x)} \right).
\end{equation*}

\medskip

Forward and backward B\"acklund--Darboux transformations are inverse to each other.

\begin{proposition}
	Suppose $\widehat{L}$ is related to $L$ by the forward B\"acklund--Darboux transformation generated by an eigenfunction $\Phi$. Then $\Phi^{-1}$ is an adjoint eigenfunction for $\widehat{L}$, and $L$ is related to $\widehat{L}$ by the backward B\"acklund--Darboux transformation generated by $\Phi^{-1}$.
\end{proposition}
\begin{proof}
	We only have to prove that $\Phi^{-1}$ is an adjoint eigenfunction for $\widehat{L}$, the rest is elementary. So we have to prove that $\d_n\Phi=\Phi^2(\widehat{L}^n_+)^\dagger(\Phi^{-1})$, which is done by the following computation:
	\begin{align*}
		\Phi^2(\widehat{L}^n_+)^\dagger(\Phi^{-1})=&\res\left(\Phi^2\circ(\widehat{L}^n_+)^\dagger\circ\Phi^{-1}\circ\d^{-1}\right)=\res\left(\Phi\circ(\widehat{L}^n_+)^\dagger\circ\Phi^{-1}\circ\d^{-1}\circ\Phi\right)=\\
		=&\res\Big(\underbrace{\Phi\circ\d^{-1}\circ\Phi^{-1}}_{=D^{-1}}\circ\widehat{L}^n_+\circ\Phi\Big)=\res\left((L^n_+\circ D^{-1}+D^{-1}\circ\d_nD\circ D^{-1})\circ\Phi\right)=\\
		=&\res\left(L^n_+\circ D^{-1}\circ\Phi\right)=\res\left(L^n_+\circ \Phi\circ\d^{-1}\right)=L^n_+\Phi=\\
		=&\d_n\Phi.
	\end{align*}
\end{proof}

\medskip

The adjunction involution interchanges forward and backward B\"acklund--Darboux transformations in the following sense.

\begin{proposition}
	Suppose $L$ is a solution of the KP hierarchy, and $M$ is obtained from~$L$ by a forward B\"acklund--Darboux transformation. Then $M^*$ is related to $L^*$ by a backward B\"acklund--Darboux transformation.
\end{proposition}
\begin{proof}
	Elementary computation.
\end{proof}

\medskip

\subsection{Example: multi-soliton solutions}

To give a concrete example, we will recall the definition of $N$-soliton solutions of the KP hierarchy and show that they are obtained from the trivial solution by $N$ successive B\"acklund--Darboux transformations.

\medskip

Let $N$ be a natural number. Choose some constants $\alpha_i$, $\beta_i$, $a_i$ for $i=1, \dots, N$. Let $y_i(\bft): = e^{\xi(\bft, \alpha_i)} + a_i e^{\xi(\bft, \beta_i)}$. For $k=1,\dots, N$, define the Wronskian determinants
\begin{equation*}
	\Delta_k: = \det \begin{pmatrix}
		y_1 & \dots & y_k \\
		y'_1 & \dots & y'_k \\
		\dots & \dots & \dots\\
		y^{(k-1)}_1 & \dots & y^{(k-1)}_k
	\end{pmatrix},
\end{equation*}
where $y^{(i)}_j:=\d^i y_j$.

\medskip

\begin{proposition}[see, e.g., Section~6.3 in~\cite{Dic03}]
	If the constants $\alpha_i$, $\beta_i$, $a_i$ satisfy $\Delta_N(0) \neq 0$, then $\Delta_N(\bft)$ is a tau-function of the KP hierarchy. The KP solution defined by $\Delta_N$ is called the \emph{$N$-soliton solution}.
\end{proposition}

\medskip

Suppose that $\Delta_k(0)\ne 0$ for all $1\le k\le N$.

\begin{proposition}
	For all $1\le k\le N$, $\Delta_k$ is a KP tau-function, obtained from $\Delta_{k-1}$ by a B\"acklund--Darboux transformation.
\end{proposition}
\begin{proof}
	As shown in \cite[Section~6.3]{Dic03}, we have
	\begin{equation*}
		\Delta_{k-1}(\bft - [x^{-1}]) = \det \begin{pmatrix}
			y_1 & \dots & y_{k-1} & x^{-k+1}\\
			y'_1 & \dots & y'_{k-1} & x^{-k+2} \\
			\dots & \dots & \dots & \dots\\
			y^{(k-1)}_1 & \dots & y^{(k-1)}_{k-1} & 1
		\end{pmatrix}.
	\end{equation*}
By direct computation, it follows that $\Delta_k(\bft) = \Coef_{x^0} \left( \gamma_{k}(x) \, \Delta_{k-1}(\bft - [x^{-1}]) \,  x^{k-1} e^{\xi(\bft, x)}  \right)$, where $	\gamma_{k}(x) = \frac{1}{1 - \alpha_k x^{-1}} + \frac{a_k}{1 - \beta_k x^{-1}}$, which completes the proof.
\end{proof}

\medskip

\subsection{B\"acklund--Darboux transformations in fermionic form}\label{subsection:bd-fermionic}

Consider a tau-function of the KP hierarchy given by $\tau = \Psi_N \Big( \bigwedge_{i\ge 1}f_{N-i}(z)\Big)$, and let $C(x) = \sum_{i \geq N} c_i x^{-i}$, where $c_i \in \mbC$ and $c_N \neq 0$. Combining formula \eqref{eq:backlund-darboux-for-tau} for a B\"acklund--Darboux transformation of a tau-function and the fermionic representation \eqref{eq:psi in terms of fi} for the wave function, we obtain the following formula:
\begin{multline*}
	\widehat{\tau}(\bft) = \Coef_{x^0} \left[C(x)\psi_{N}(\bft,x)\tau(\bft) \right] 
	= \Coef_{x^0} \left[ \Psi_{N+1}  \Big( \sum_{j \geq N} c_j x^{-j} \sum_{k \in \mbZ} \theta_k x^{k} \bigwedge_{i\ge 1}f_{N-i}(z) \Big) \right]= \\
	=\Psi_{N+1} \Big( \sum_{j \geq N} c_j \theta_j \bigwedge_{i\ge 1}f_{N-i}(z)\Big) = \Psi_{N+1} \Big( \bigwedge_{i\ge 1}f_{N+1-i}(z)\Big),
\end{multline*}
where we set $f_N(z) = C(z^{-1})$. Therefore, if $\tau = \tau_H$ for $H \in \Gr_N^0$, then $\widehat{\tau} = \tau_{\widehat{H}}$, where $H \subset \widehat{H} := H \oplus \spn(f_N) \in \Gr_{N+1}^0$. 

\medskip

In other words, in the Sato Grassmannian formalism, a forward B\"acklund--Darboux transformation corresponds to an extension of the infinite-dimensional space $H$ by some one-dimensional subspace.

\medskip

Dually, since a backward B\"acklund--Darboux transformation can be represented using the adjoint involution as $\widecheck{\tau} = (\widehat{\,\tau^*\,})^*$, we obtain $\widecheck{\, \tau_H \,} = \tau_{\widecheck{H}}$, where $H \supset \widecheck{H} := H \cap \ell^\perp$ for some one-dimensional subspace $\ell \subset \mbC((z))$. Thus, a backward B\"acklund--Darboux transformation restricts the subspace $H$ to its codimension one subspace.

\medskip

It follows that two KP solutions are related by some sequence of forward and backward B\"acklund--Darboux transformations if and only if the corresponding subspaces of $\mbC((z))$ are commensurable.

\medskip

\begin{remark}{\ }
\begin{itemize}
\item Similarly to Remark~\ref{remark:about base algebra}, the theory of B\"acklund--Darboux transformations has a direct generalization to the case of arbitrary base algebra~$K$. The necessary changes are minimal: in Propositions~\ref{proposition:bd-definition} and~\ref{proposition:backlund-for-tau} one should require that~$\Phi(0)$ is invertible, and in Proposition~\ref{proposition:backward bd} one should require that~$\bPhi(0)$ is invertible. 

\smallskip

\item Consider an arbitrary tau-function $\tau$ of the KP hierarchy, with the base algebra~$K$. By Proposition \ref{proposition:wave-is-eigenfunction}, the wave function $\psi(\bft, x)$ can be viewed as an eigenfunction of the corresponding solution, if we extend the base algebra to $K((x^{-1}))$. Note that $\psi(0,x)=\frac{\tau(-[x^{-1}])}{\tau(0)}\in K((x^{-1}))$ is an invertible element. Therefore, $\tau(\bft) \psi(\bft, x) = \tau(\bft - [x^{-1}])e^{\xi(\bft, x)}$ is a KP tau-function, which is obtained from $\tau$ by a B\"acklund--Darboux transformation.
\end{itemize}
\end{remark}

\medskip

\subsection{MKP hierarchy}

\begin{definition}
A \emph{tau-sequence of the mKP (modified KP) hierarchy} is a sequence $(\tau_N)_{N \in \mathbb{Z}}$ of KP tau-functions such that for each $N \in \mathbb{Z}$, the tau-function $\tau_{N+1}$ is obtained from $\tau_N$ via a forward B\"acklund--Darboux transformation.
\end{definition}

\medskip
 
It follows from the discussion in Section \ref{subsection:bd-fermionic} that an mKP tau-sequence $(\tau_N)$ can be represented by a point $H_{\bullet}$ in the big cell of flag manifold $\mathrm{Fl}^0$, that is, an infinite sequence of subspaces $( \ldots \subset H_{-1} \subset H_0 \subset H_1 \subset \ldots)$ with $H_N \in \Gr_N^0$.

\medskip

Equivalently, given a sequence of Laurent series $(f_N)_{N \in \mbZ}$ of the form
$$
	f_{N}(z)=z^{N} \Big(1 + \sum_{j \geq 1} f_{N,j} z^j\Big), \quad f_{N,j} \in \mbC,
$$
we define the corresponding an mKP tau-sequence $(\tau_N)_{N \in \mbZ}$ by
\begin{equation}
	\label{eq:mkp-sequence-from-series}
	\tau_{N} := \Psi_N \Big( \bigwedge_{i\ge 1}f_{N-i}(z)\Big) \in \mbC[[\bft]].
\end{equation}

\end{document}